\documentclass[reprint,aps,superscriptaddress,twocolumn]{revtex4-1}

\usepackage{amsmath,amssymb,amsthm}

\usepackage{mathtools}
\mathtoolsset{centercolon}

\usepackage{bbold}

\usepackage{tikz}
\usetikzlibrary{calc}

\usepackage{hyperref}

\hypersetup{
  colorlinks = true
}

\definecolor{darkred}  {rgb}{0.5,0,0}
\definecolor{darkblue} {rgb}{0,0,0.5}
\definecolor{darkgreen}{rgb}{0,0.5,0}

\hypersetup{
  urlcolor  = darkblue,         
  linkcolor = darkblue,     
  citecolor = darkgreen,    
  filecolor = darkred       
}

\newtheorem*{theorem}{Theorem}
\newtheorem{lemma}{Lemma}

\newcommand{\ket}[1]{|#1\rangle}
\newcommand{\bra}[1]{\langle#1|}
\newcommand{\ketbra}[2]{|#1\rangle\langle#2|}

\newcommand{\proj}[1]{\ketbra{#1}{#1}}

\def\>{\rangle}
\def\<{\langle}

\newcommand{\ct}{^{\dagger}}         

\newcommand{\x}{\otimes}

\DeclareMathOperator{\tr}{Tr}

\DeclarePairedDelimiter{\set}{\lbrace}{\rbrace}
\DeclarePairedDelimiter{\abs}{\lvert}{\rvert}
\DeclarePairedDelimiter{\norm}{\lVert}{\rVert}

\DeclarePairedDelimiter{\ceil}{\lceil}{\rceil}
\DeclarePairedDelimiter{\of}{\lparen}{\rparen}
\DeclarePairedDelimiter{\ob}{\lbrack}{\rbrack}

\newcommand{\mc}[1]{\mathcal{#1}}
\newcommand{\N}{\mc{N}}
\newcommand{\E}{\mc{E}}

\newcommand{\Ic}{I_\mathrm{coh}}

\newcommand{\reg}[1]{\mathsf{#1}}

\usepackage{hyperref}

\hypersetup{
  colorlinks = true
}

\definecolor{darkred}  {rgb}{0.5,0,0}
\definecolor{darkblue} {rgb}{0,0,0.5}
\definecolor{darkgreen}{rgb}{0,0.5,0}

\hypersetup{
  urlcolor  = darkblue,         
  linkcolor = darkblue,     
  citecolor = darkgreen,    
  filecolor = darkred       
}

\def\Eflag{\reg{F}}
\def\Gflag{\reg{G}}

\newcommand{\trans}[2]{#2 \to #1}
\def\A{\reg{A}}
\def\B{\reg{B}}
\def\1{\mathbb{1}}

\def\SWin{\reg{S}}
\def\SWout{\reg{S}}
\def\Min{\reg{A}}
\def\Mout{\reg{B}}
\def\R{\reg{R}}

\def\chan{\mathcal{M}}
\def\ngam{\tilde{\Gamma}_{\kappa}} 
\def\gam{\Gamma}
\def\id{\mathcal{I}}
\def\Afull{\tilde{\reg{A}}}
\def\Bfull{\tilde{\reg{B}}}
\def\Ashield{\reg{A}}
\def\mms{\mu}   
\def\mes{\phi}  
\def\sap{\mms}
\def\Qsys{\reg{Q}}
\def\Akey{\reg{a}}
\def\Bshield{\reg{B}}
\def\Bkey{\reg{b}}
\def\key{\Akey\Bkey}
\def\shield{\Ashield\Bshield}
\def\A{\reg{A}}
\def\B{\reg{B}}
\def\Q{\reg{Q}}
\def\Eflag{\reg{F}}
\def\Gflag{\reg{G}}
\def\block{\reg{C}}

\def\NN{N}
\def\r{r}
\def\m{m}
\def\q{q}

\def\Ox{\bigotimes}
\def\HHHOstate{\zeta}
\def\GammaChoi{\HHHOstate}

\begin{document}

\title{Unbounded number of channel uses are required to see quantum capacity}

\newcommand{\DAMTP}{Department of Applied Mathematics and Theoretical Physics, University of Cambridge, Cambridge CB3 0WA, U.K.}

\author{Toby Cubitt}
\affiliation{\DAMTP}
\author{David Elkouss}
\affiliation{Departamento de An\'alisis Matem\'atico and Instituto de Matem\'atica Interdisciplinar, Universidad Complutense de Madrid, 28040 Madrid, Spain}
\author{William Matthews}
\affiliation{\DAMTP}
\affiliation{Statistical Laboratory, University of Cambridge, Wilberforce Road, Cambridge CB3 0WB, U.K.}
\author{Maris Ozols}
\affiliation{\DAMTP}
\author{David P\'erez-Garc\'ia}
\affiliation{Departamento de An\'alisis Matem\'atico and Instituto de Matem\'atica Interdisciplinar, Universidad Complutense de Madrid, 28040 Madrid, Spain}
\author{Sergii Strelchuk}
\affiliation{\DAMTP}


\date{\today}

\maketitle

{\bf Transmitting data reliably over noisy communication channels is one of the most important applications of information theory, and well understood when the channel is accurately modelled by classical physics. However, when quantum effects are involved, we do not know how to compute channel capacities.
The capacity to transmit quantum information is essential to quantum cryptography and computing, but the formula involves maximising the coherent information over arbitrarily many channel uses~\cite{Lloyd_97,Shor_02,Devetak_05}. This is because entanglement across channel uses can increase the coherent information~\cite{DiVincenzo_98}, even from zero to non-zero~\cite{Smith_08}! However, in all known examples, at least to detect whether the capacity is \emph{non-zero}, two channel uses already suffice~\cite{Smith_08,Brandao_12}. Maybe a finite number of channel uses is always sufficient? Here, we show this is emphatically not the case: for any $n$, there are channels for which the coherent information is zero for $n$ uses, but which nonetheless have capacity. This may be a first indication that the quantum capacity is uncomputable.}

In the classical case, not only can we exactly characterise the maximum rate of communication over any channel -- its \emph{capacity} -- we also have practical error-correcting codes that attain this theoretical limit.  
It is instructive to review why the capacity of classical channels is a solved problem. 
 Even though optimal communication over a discrete, memoryless classical channel involves encoding the information across many uses of the channel, Shannon showed that a channel's capacity is given mathematically by optimising an entropic quantity (the \emph{mutual information}) over a \emph{single} use of the channel. 
This follows immediately from the fact that mutual information is \emph{additive}. Thanks to this, the capacity of any classical channel can 
be computed efficiently.

It is for this reason that additivity questions for quantum channel capacities took on such importance, and why the major recent breakthroughs proving that additivity is violated~\cite{Hastings,Smith_08} had such an impact. 
It has been known for some time~\cite{Lloyd_97,Shor_02,Devetak_05} that the quantum capacity is given by a \emph{regularised} expression---the optimisation of an entropic quantity (the \emph{coherent information}) in the limit of arbitrarily many uses of the channel:
\begin{align*}\label{eq:Q}
  Q^{(n)}(\N) &:= \frac{1}{n} \max_{\rho^{(n)}} \Ic(\N^{\otimes n}, \rho^{(n)}),\\
  Q(\N) &:= \lim_{n\to\infty} Q^{(n)}(\N),
\end{align*}
(Here, $Q^{(n)}(\N)$ is the coherent information $\Ic$ maximised over a joint input $\rho^{(n)}$ for $n$ uses of the channel $\N$.) However, the regularisation renders computing the quantum capacity infeasible; it involves an optimisation over an infinite parameter space.

Were the coherent information additive (i.e.\ $Q^{(n)}(\N) = Q^{(1)}(\N)$), the regularisation could be removed and the quantum capacity could be computed as easily as the capacity of classical channels. However, this is not the case. The first explicit examples of this \emph{superadditivity} phenomenon were given by Di~Vincenzo et al.~\cite{DiVincenzo_98}, and extended by Smith et al.~\cite{Smith_07}. For these examples (where $\mc{N}$ is a particular depolarising channel) it was shown (numerically) that $0 \leq Q^{(1)}(\N) < Q^{(n)}(\N)$ for small values of $n\le 33$.

While the \emph{classical} capacity of quantum channels also involves a regularised formula~\cite{Hastings}, we at least know precisely in which cases it is \emph{zero}: simply for those channels whose output is completely independent of the input. The set of zero-\emph{quantum}-capacity channels is much richer. Indeed, we do not even have a complete characterisation of \emph{which} channels have zero quantum capacity. To date, we know of only two kinds of channels with zero quantum capacity: \emph{antidegradable} channels~\cite{BDS97,CRS08} and \emph{entanglement-binding} channels~\cite{HHH00}. The former has the property that the environment can reproduce the output, thus $Q = 0$ by the no-cloning theorem~\cite{WZ82}. The latter can only distribute PPT entanglement, which cannot be distilled by local operations and classical communication~\cite{PPT}, which again implies $Q = 0$.

This has dramatic consequences. It is possible to take two quantum channels above, $\N_1$ antidegradable and $\N_2$ entanglement-binding, which individually have no capacity whatsoever, yet when used together can transmit quantum information reliably ($Q(\N_1 \otimes \N_2) > 0$). This \emph{superactivation} phenomenon was discovered recently by Smith and Yard~\cite{Smith_08}. They also used their examples to construct a single channel $\N$ exhibiting an extreme form of superadditivity of the coherent information, where $0 = Q^{(1)}(\N) < Q^{(2)}(\N)$. (In their construction, having two uses of $\N$ effectively enables one use of $\N_1$ and one of $\N_2$.) Even stronger superactivation phenomena have been shown in the context of zero-error communication over quantum channels~\cite{Alon,Cubitt,Cubitt!,Cubitt!!,Shirokov_14}.

On the one hand, additivity violation means regularisation is required in formulas for computing capacities. On the other hand, it also means that \emph{entanglement} can protect information from noise (the coherent information is additive for unentangled input states). 

But just how bad can this additivity violation be? One might hope that, at least in determining whether the quantum capacity is non-zero, one need only consider a finite number of uses of a channel. Indeed, since the Smith and Yard construction relies on combining the only two known types of zero-capacity channels, one might dare to hope that even two uses suffice.
%
(Similarly, for the classical capacity of quantum channels the only known method for constructing examples of additivity violation~\cite{HaydenWinter,Hastings} cannot give a violation for more than two uses of a channel, and there is some evidence that this may be more than just a limitation of the proof techniques~\cite{Ashley}.)
Was this indeed the case, additivity violation would be reduced to something relatively benign: entangling the inputs across more than two uses of the channel would give no advantage. And one would be able to compute the quantum capacity by optimising the coherent information over two uses of the channel, which is not substantially more difficult than the optimisation over a single channel use.

In this paper, we show for the first time that this is not the case: additivity violation is as bad as it could possible be. We prove that, for any $n$, one can construct a channel $\N$ for which the coherent information of $n$ uses is zero ($Q^{(n)}(\N)=0$), yet for a larger number of uses the coherent information is strictly positive, implying that the channel has non-zero quantum capacity ($Q(\mathcal N)>0$). This is also the first proof
that there can be a gap between $Q^{(n)}(\N)$ and the quantum capacity for an arbitrarily large number $n$ of uses of the channel. Our result implies that, in general, one must consider an arbitrarily large number of uses of the channel \emph{just to decide whether the channel has any quantum capacity at all!}

Perhaps the earliest indication that determining the quantum capacity may be a difficult problem comes from the work of Watrous~\cite{Watrous_04}, who showed that an arbitrarily large number of copies of a bipartite quantum state can be required for entanglement distillation assisted by two-way classical communication. Our result can be regarded as the counterpart of~\cite{Watrous_04} for the quantum capacity (which is mathematically equivalent to entanglement distillation assisted by \emph{one-way} communication). However, the proof ideas and techniques of~\cite{Watrous_04} require two-way communication, thus they do not apply to the usual capacity setting. Our result is instead based on the ideas of Smith and Yard, in particular the intuition provided by Oppenheim's commentary thereon~\cite{Oppenheim_08}.



This intuition comes from a class of bipartite quantum states called \emph{pbits} (private bits)~\cite{Horodecki_09b}: $\rho_{\Akey\Ashield\Bkey\Bshield} = \tfrac{1}{2}(\proj{\phi^+}_{\key}\otimes\sigma^+_{\shield} + \proj{\phi^-}_{\key}\otimes\sigma^-_{\shield})$, together with the standard equivalences between quantum capacity (sending entanglement over a channel) and distilling entanglement from the Choi-Jamio\l{}kowsky state associated with the channel. Here, $\ket{\phi^{\pm}}$ are Bell states, and $\sigma^{\pm}$ are hiding states~\cite{Eggeling_02}. The latter are orthogonal (globally perfectly distinguishable), but cannot be distinguished using local operations and classical communication (LOCC).

If $\rho_{\Akey\Ashield\Bkey\Bshield}$ is shared between Alice (who holds $\Akey\Ashield$) and Bob (who holds $\Bkey\Bshield$), then they share at least one ebit of entanglement due to the Bell states. But this entanglement is inaccessible to them unless they can determine which of the two Bell states they share. This they could do if only they could determine which hiding state they have. But $\sigma^{\pm}$ cannot be distinguished by LOCC, preventing them from extracting the entanglement from $\rho_{\Akey\Ashield\Bkey\Bshield}$. The $\key$ part of the system is usually called the ``key'', and $\shield$ the ``shield'' (as it decouples the systems $\key$ from any external system).

Now imagine they have access to a quantum erasure channel $\E_{\frac{1}{2}}$, which with probability $1/2$ transmits its input perfectly, and with probability $1/2$ completely erases it. It is well known that such a channel cannot be used to transmit any entanglement. However, if they also share $\rho_{\Akey\Ashield\Bkey\Bshield}$, Alice can use the erasure channel to send her part $\Ashield$ of the shield to Bob. If the erasure channel transmits, Bob now holds the entire $\shield$ system and can now distinguish $\sigma^{\pm}$. Thus, with probability $1/2$, Alice and Bob can now extract the entanglement from $\rho_{\Akey\Ashield\Bkey\Bshield}$.

Instead of supplying Alice and Bob with the state $\rho_{\Akey\Ashield\Bkey\Bshield}$ and an erasure channel, we instead supply them with a \emph{switched} channel. This has an auxiliary classical input that controls whether the channel acts as $\E_{\frac{1}{2}}$ or $\Gamma$, where $\Gamma$ is the channel with Choi-Jamio\l{}kowsky state $\rho_{\Akey\Ashield\Bkey\Bshield}$. The above argument then implies that no quantum information can be sent over a single use of the channel, but it can be sent using two uses, by switching one to $\E_{\frac{1}{2}}$ and the other to $\Gamma$.

This is the intuition behind the Smith and Yard construction~\cite{Oppenheim_08}. However, because it is constructed out of two very particular types of quantum channels, this idea does not seem to extend to larger numbers of uses. Nonetheless, the intuition behind our result is based on a refinement of these ideas, which we now sketch.

We want to achieve two seemingly contradictory goals: (1)~To prevent Alice from sending any quantum information to Bob over $n$ of uses of the channel. (2)~To permit this when Alice has access to some larger number of uses $N>n$. We can achieve~(1) by increasing the erasure probability of the erasure channel to something much closer to~1, and also adding noise to the $\Gamma$ channel; the noise then swamps any entanglement. The problem is that this seems to also render~(2) impossible. If the channel is so noisy that it destroys all entanglement sent through it, then (by definition) no amount of coding over multiple uses of the channel can succeed in transmitting quantum information.

However, note that the information that Alice needs to send to Bob in order to extract entanglement from the pbit $\rho_{\Akey\Ashield\Bkey\Bshield}$ is essentially classical. Bob just needs to know one classical bit of information to distinguish the two hiding states. This suggests that \emph{classical} error correction might help Alice send this information to Bob, even when the channel is very noisy. The intuition behind our proof is that a simple classical repetition code suffices. Instead of the pbit $\rho_{\Akey\Ashield\Bkey\Bshield}$, we use a pbit $\tfrac{1}{2}(\proj{\phi^+}_{\key}\otimes\sigma^+_{\A_1\B_1}\otimes\cdots\otimes\sigma^+_{\A_N\B_N} + \proj{\phi^-}_{\key}\otimes\sigma^-_{\A_1\B_1}\otimes\cdots\otimes\sigma^-_{\A_N\B_N})$ that contains $N$ 
copies of the shield. For Bob to distinguish the hiding states, it suffices for one copy to make it through the erasure channel. Alice now tries to send all of the copies of the shield through many uses of the erasure channel. However high the erasure probability, the probability that at least one will get through can be made arbitrarily high for sufficiently many attempts.

We now give a more precise description of our construction. The \emph{erasure channel with erasure probability} $p$ is $\E_p^{\trans{\Eflag\B}{\A}} := (1-p) \ketbra{0}{0}^{\Eflag} \x \id^{\trans{\B}{\A}} + p \ketbra{1}{1}^{\Eflag}\x \1^{\B}/\dim (\B)$, where $\id^{\trans{\B}{\A}}$ is the identity channel from $\A$ to $\B$, and $\Eflag$ is the erasure flag. The channel $\Gamma^{\trans{\Bfull}{\Afull}}$ belongs to the class of PPT entanglement-binding channels whose Choi state is an approximate pbit~\cite{Horodecki_09b}. We show that $\Gamma$ can be constructed with $\A := \A_1 \ldots \A_N$ and $\B := \B_1 \ldots \B_N$ consisting of $N$ parts, such that even if Bob only receives part $\A_i$ of Alice's shield for any $i$, they obtain approximately one ebit of one-way distillable entanglement. Let $\ngam^{\trans{\Eflag\Bfull}{\Afull}} := \E^{\trans{\Eflag\Bfull}{\Bfull}}_{\kappa}\circ \gam^{\trans{\Bfull}{\Afull}} $ be a noisy version of the channel $\Gamma$. Our construction uses channels of the form
\begin{equation}\label{eq:ourchannel}
  \chan^{\trans{\SWout \Eflag\Bfull}{\SWin\Afull}}
  := \mc{P}_0^{\trans{\SWout}{\SWin}}
            \otimes \ngam^{\trans{\Eflag\Bfull}{\Afull}}
   + \mc{P}_1^{\trans{\SWout}{\SWin}}
            \otimes \E_p^{\trans{\Eflag\Bfull}{\Afull}}.
\end{equation}
Here $\mc{P}_i^{\trans{\SWout}{\SWin}}$ projects onto the $i$-th computational basis vector of the qubit system $\SWin$ which thereby acts
as a classical switch allowing Alice to choose whether the channel
acts as $\E_p$ or $\tilde{\Gamma}_{\kappa}$ on the main input $\Afull$.
$\SWin$ is retained in the output which lets Bob learn which choice was made.

Making the above intuition rigorous for this channel is non-trivial: First, we must prove that the coherent information of $n$ uses of the channel is strictly zero, for \emph{any} input to the channel (not just the input states from the above intuition). To this end, we cannot just directly use a pbit with $N$-copy shield of the form given above, as it would have distillable entanglement. Fortunately, we find that an approximate pbit construction from \cite{Horodecki_09b} can be adapted for the role. But then we must take this approximation into account in the proof that the channel \emph{does} have capacity. This requires a careful analysis of the various parameters of our channel to show that both of the desired properties can hold simultaneously, which requires a somewhat delicate argument. The technical arguments are described in the Methods section.

One natural question (which we leave open) is whether we can obtain a stronger form of our result with a constant upper bound on the channel dimension. It would also be interesting to see if one can obtain a result analogous to ours for the \emph{private} capacity of quantum channels. Finally, our result gives a first indication that the quantum capacity of a channel might well be an uncomputable quantity; uncomputability of the quantum capacity would necessarily imply the behaviour we have shown here.

\vspace{-1em}

\section*{Methods}
\vspace{-1em}
We state and outline the proof of our main result -- for any number of uses we can show that there exists a channel with positive capacity but zero coherent information. Formally, we prove the following:
\begin{theorem}
Let $\chan$ be the channel defined in Eq.~\eqref{eq:ourchannel}.
For any positive integer $n$, if $\kappa \in (0,1/2)$ and $p \in [(1+\kappa^{n})^{-1/n},1]$ then we can choose $\NN$ and $\Gamma$ such that:
\begin{enumerate}
  \item $Q^{(n)}(\chan) = 0$ and
  \item $Q^{(\NN+1)}(\chan) > 0$, and therefore $Q(\chan) > 0$.
\end{enumerate}
\end{theorem}

 The proof is divided in two parts. We first prove that, given $n$ and $\kappa$, for any $\Gamma$ with zero capacity there is a range of $p$ that makes the coherent information of $\chan^{\x n}$ zero. In the second part we prove that there exists $\Gamma$ with zero capacity
such that $\chan$ has positive capacity.

%

For the first part we can simplify the analysis of $\chan^{\x n}$
by showing that it is optimal to make a definite choice
(i.e.~a computational basis state input) for each of the
$n$ switch registers.
For each possible setting of the $n$ switches,
the
coherent information is a convex combination
of the coherent information
for three
cases, weighted by their
probabilities: (a)~every channel erases, (b)~all of the $\E_p$
erase but not all $\tilde{\Gamma}$ erase, (c)~at least one of the
$\E_p$ does not erase (and therefore acts as the identity channel).
The coherent information for cases $(b)$ and $(c)$ can be upper bounded respectively by zero and $H(\reg{R})$, where $\reg{R}$ is a system that purifies the input. For $(a)$ it is bounded above by $-H(\reg{R})$.
Weighting by the probabilities,
we find that the total coherent information is upper-bounded
by $\of[\big]{1 - (1+\kappa^n) p^n} H(\reg{R})$.
This allows us to conclude that for any $n$ and $\kappa$ we can find $p$ such that the coherent information of $n$ uses of the channel is zero.

%

To prove the second part, we show that for fixed $\kappa,p$ we can find a $\Gamma$ with 
an $N$-copy shield
such that the coherent information of $N+1$ uses of the channel $\chan$ is positive for some $N+1>n$. We number the
channel uses $0,\ldots,N$ and label the
systems involved in the $i$-th use of the channel with superscript $i$.
Consider the following input.
The switch registers
are set to choose $\tilde{\Gamma}_{\kappa}$ for use $0$ and $\mathcal E_p$ for the remaining uses $1, \ldots, N$.
We maximally entangle subsystem $\A^{0}_i$ of $\Afull^{0}$ (which is
acted on by $\tilde{\Gamma}_{\kappa}$) with subsystem $\A^{i}_1$ of $\Afull^{i}$ (acted on
by an erasure channel). We also maximally entangle subsystem
$\Akey^{0}$ of $\Afull^{0}$ with a purifying reference system
$\Akey$ which is retained
by Alice. The remaining input subsystems are set to an arbitrary pure state.
The resulting coherent information is a convex combination of cases
where
(a) $\ngam$ 
erases,
(b) $\ngam$ 
does not erase
but all the $\mc{E}_p$ erase, and (c)
$\ngam$ 
and at least one $\mc{E}_p$ do not erase.
Case (a) contributes coherent information $-1$ weighted by
its probability $\kappa$. Case (b) contributes approximately
zero coherent information (due to a standard property of pbits).
In case (c),
after channel use $0$, Alice and Bob
share the Choi state of $\Gamma$ on systems
$\Akey \Bkey^{0} \A^{1}_1 \B^{0}_{1} \ldots \A^{N}_1 \B^{0}_{N}$,
and after the $N$ uses of $\mc{E}_p$
at least one of $\A^{1}_1 \ldots \A^{N}_1$ reaches Bob unerased.
They then share a state with approximately one ebit
of one-way distillable entanglement (coherent information $+1$).
This contribution is weighted by the probability $(1-\kappa)(1-p^N)$.
We show that for $p\in (0,1)$, $\kappa \in (0,1/2)$,
we can find a $\Gamma$ with large enough $N$
for which the overall coherent information
is positive, proving that $Q(\chan) > 0$.
Further mathematical details are given in the Supplementary Information.
\mbox{}

{\bf Acknowledgements:} DE and DP acknowledge financial support from the European CHIST-ERA project CQC (funded partially by MINECO grant PRI-PIMCHI-2011-1071). TSC is supported by the Royal Society.
MO acknowledges financial support from European Union under project QALGO (Grant Agreement No. 600700).



\begin{widetext}
\appendix
\def\bibsection{\section*{SUPPLEMENTARY REFERENCES}}

\makeatletter
\def\tagform@#1{\maketag@@@{(S\ignorespaces#1\unskip\@@italiccorr)}}
\makeatother

\vspace{-1cm}
\setcounter{equation}{0}
\section*{SUPPLEMENTARY INFORMATION}

\subsection*{Preliminaries}


In the following, each system $\Q$ is associated to a Hilbert
space of \emph{finite} dimension $\dim(\Q)$, and the Hilbert space
has an orthonormal \emph{computational} basis
$\{ \ket{i}^{\Q} : i \in \{0, \dotsc, \dim(\Q)-1 \} \}$.
For any system $\Qsys$, let $\mms^{\Qsys} := \1^{\Qsys}/\dim(\Qsys)$
denote its maximally mixed state.
Let $\A$ and $\B$ be two systems of equal dimension, $\id^{\trans{\B} {\A}}$ denote the identity channel between them, and $\Eflag$ be a binary \emph{erasure flag}.
The \emph{total erasure channel} $\E_1^{\trans{\Eflag\B} {\A}}$ maps any input state to $\ketbra{1}{1}^{\Eflag}\x \mu^{\B}$, while
$\E_p^{\trans{\Eflag\B}{\A}}
:= (1-p) \ketbra{0}{0}^{\Eflag} \x \id^{\trans{\B}{\A}} + p \E_1^{\trans{\Eflag\B}{\A}}$
denotes the erasure channel with erasure probability $p$.
For any number of uses of 
$\E_1$ and any input state $\rho$ we have
\begin{equation}
    \Ic(\mc{E}^{\x n}_1,\rho) = -S(\rho).
    \label{total-erasure}
\end{equation}

For any register $\Eflag$, a \emph{flagged channel} is of the form
$\N^{\trans{\Eflag\B}{\A}}
= \sum_{i = 0}^{\dim(\Eflag)-1} p_i \ketbra{i}{i}^{\Eflag} \x \N_i^{\trans{\B}{\A}}$.
An example is $\E_p^{\trans{\Eflag\B}{\A}}$. For any flagged channel we have
\begin{equation}
  \Ic(\N^{{\Eflag\B}{\A}},\rho^{\A})
= \sum_{i}
p_i \Ic(\N_i^{\trans{\B}{\A}},\rho^{\A}),\label{flagged-coh}
\end{equation}
which follows easily from
\begin{equation}
I(\R\rangle\B\Eflag)_{\sum_{i}
p_i \rho^{\R\B}_i \x \ketbra{i}{i}^{\Eflag}}
= \sum_{i}
p_i I(\R\rangle\B)_{\rho^{\R\B}_i}\,.\label{flagged-state-coh}
\end{equation}

For any $i \in \{0, \dotsc, \dim(\SWin)-1 \}$,
let $\mc{P}_i^{\trans{\SWout}{\SWin}}$
denote the
completely positive map
$X^{\SWin} \mapsto \proj{i}^{\SWout} X^{\SWin} \proj{i}^{\SWout}$.
A \emph{switched 
channel} is a 
channel of the form
$
  \sum_{i = 0}^{\dim(\SWin)-1}
  \mc{P}_i^{\trans{\SWout}{\SWin}} \x \N_i^{\trans{\Mout}{\Min}}
$
where each $\N_i$ is a quantum channel.
The register $\SWin$ acts as a classical
switch allowing the sender to choose between
different channels $\N_i$
to be applied on the ``main input'' $\Min$
to produce a state of the ``main output'' $\Mout$.
We will need the following simple lemma regarding switched channels:
\begin{lemma}
For any switched channel,
\begin{equation}
    \max_{\rho^{\SWin\A}}
    \Ic(\N^{\trans{\SWout \Mout}{\SWin \Min}},\rho^{\SWin\A})\\
    =
    \max_i 
    \max_{\rho^{\A}}
    \Ic(\N_i^{\trans{\B}{\A}},\rho^{\A}) \label{switched-coh}
\end{equation}
where $0\leq i \leq \dim(\SWin)-1$.
\end{lemma}
\begin{proof}
To see this, note that any purification $\rho^{\SWin\A\R}$
of $\rho^{\SWin\A}$ can be written in the form
\begin{equation}\label{switch-form}
    \ket{\rho^{\SWin\A\R}}
    = \sum_{i} \sqrt{p_i} \ket{i}^{\SWin} \x \ket{\rho_i}^{\A\R}.
\end{equation}
Here $p_i$ is the probability that the switch is set to $i$,
and $\ket{\rho_i}^{\A\R}$ is a purification of the
channel input state $\rho_i^{\A}$ conditioned on that setting.
Conversely, given probabilities $p_i$ and states
$\rho_i^{\A}$ for each switch value, we can always find
$\ket{\rho^{\SWin\A\R}}$ satisfying \eqref{switch-form}.
Given this, we see that
\begin{equation}
    \N^{\trans{\SWout \Mout}{\SWin \Min}} (\rho^{\SWin\A\R})
    = \sum_i p_i \proj{i}^{\SWin} \x
    \N_i^{\trans{\Mout} {\Min}}(\rho_i^{\A\R})
\end{equation}
where $\rho_i^{\A\R} := \ketbra{\rho_i}{\rho_i}^{\A\R}$.
From \eqref{flagged-coh} it follows that
\begin{align}
    \Ic(\N^{\trans{\SWout \Mout}{\SWin \Min}},\rho^{\SWin\A})
    &=
    \sum_i p_i \Ic(\N_i^{\trans{\Mout} {\Min}},\rho_i^{\A})\\
    &\leq
    \sum_i p_i \max_{\rho_i^{\A}} \Ic(\N_i^{\trans{\Mout}{\Min}},\rho_i^{\A})\\
    &\leq
    \max_i \max_{\rho_i^{\A}} \Ic(\N_i^{\trans{\Mout}{\Min}},\rho_i^{\A})
\end{align}
which completes the proof.
\end{proof}


We will also require some basic facts about
\emph{pbits} (``private bits'')~\cite{Horodecki_09b},
which we gather here.
Given a bipartite system $\key$ with $\dim \Akey = \dim \Bkey = 2$
and a bipartite system $\A \B$ with $\dim \A = \dim \B$,
a \emph{perfect pbit} with \emph{key} $\key$ and \emph{shield} $\A\B$
is a state $\gamma^{\key\A \B}$ of the form
\begin{equation}
  \gamma^{\key\shield} := U^{\key\shield}
   \of[\big]{\mes^{\key} \x \sigma^{\shield}} (U\ct)^{\key\shield},
\end{equation}
where $\mes^{\key}$ is the projector onto
$\ket{\mes}^{\key} := \frac{1}{\sqrt{2}} \of{\ket{00} + \ket{11}}^{\key}$,
$\sigma^{\shield}$ is some mixed state, and
\begin{equation}
  U^{\key\shield} :=
  \sum_{i,j=0}^{1} \proj{i}^{\Akey}\x \proj{j}^{\Bkey} \x U^{\shield}_{ij}
\end{equation}
is a \emph{twisting unitary} controlled by the key $\key$
and acting on the shield $\shield$ as some unitary $U^{\shield}_{ij}$.
Note that due to
the form of $\mes^{\key}$ and $U^{\key\shield}$,
we have
\begin{equation}
    \gamma^{\key\shield}
    = \frac{1}{2} \sum_{k,l=0}^1
      \ket{k,k}^{\Akey\Bkey}
      \bra{l,l}^{\Akey\Bkey}
      \x U_{kk}^{\shield} \sigma^{\shield} (U_{ll}\ct)^{\shield}.
\end{equation}
Let us define
$U^{\Bkey\shield} := \sum_{j=0}^{1} \proj{j}^{\Bkey} \x U^{\shield}_{jj}$.
If Bob has access to $\Bkey$ and the whole shield $\shield$
then he can apply the unitary operation
$(U\ct)^{\Bkey\shield}$ to these systems,
yielding a 2-qubit maximally entangled state on $\key$. Therefore,
\begin{equation}\label{pbit-prop-1}
    I(\Akey\rangle\Bkey\shield)_{\gamma^{\key\shield}}
    = I(\Akey\rangle\Bkey\shield)_{\mes^{\key}\x\sigma^{\shield}} = 1.
\end{equation}

On the other hand, if we throw away the shield systems $\shield$,
we are left with a state $\gamma^{\key}$ that can be converted into a perfectly random shared classical bit by locally measuring systems $\Akey$ and $\Bkey$ in the standard basis.  The coherent information of a shared random bit is zero, so from \eqref{Ic-comp} we get
$I(\Akey\rangle\Bkey)_{\gamma^{\key}} \geq 0$ and thus
\begin{equation}\label{pbit-prop-2}
    I(\Akey\rangle\Bkey\shield)_{\gamma^{\key\Bshield_1} \x \mms^{\Ashield_1}} \geq 0.
\end{equation}


\subsection*{Channel construction}

We will now describe the input and output systems of our channel $\chan$.
Let $\Akey$ and $\Bkey$ be two-dimensional systems (qubits).
We call $\Akey\Bkey$ the ``key''.
Let $\A_{i,j,k}$ and $\B_{i,j,k}$
be $d$-dimensional systems
for all $i \in [\NN]$, $j \in [\r]$, $k \in [\m]$
where $[n] := \set{1, \dotsc, n}$.
We define composite systems
$\A_i := \set{\A_{i,j,k} : j \in [\r], k \in [\m]}$ and
$\Ashield := \set{\A_i : i \in [\NN]}$ for Alice,
and similar systems $\B_i$ and $\Bshield$ for Bob.
We call $\Ashield\Bshield$ the ``shield''
and call $\A_i$ ``Alice's $i$-th share of the shield''.
Let $\Eflag$ be a qubit called ``the erasure flag''.
Let $\Afull := \Akey\Ashield$, and $\Bfull := \Bkey\Bshield$.


Our construction is a switched channel
\begin{equation}\label{eq:ourchannel_again}
  \chan^{\trans{\SWout \Eflag\Bfull} {\SWin\Afull}}
  := \mc{P}_0^{\trans{\SWout}{\SWin}}
            \otimes \ngam^{\trans{\Eflag\Bfull}{\Afull}}
   + \mc{P}_1^{\trans{\SWout} {\SWin}}
            \otimes \E_p^{\trans{\Eflag\Bfull}{\Afull}}.
\end{equation}
It depends on parameters $\NN,r,m \in \mathbb{N}$
and $p, \kappa, q \in [0,1]$, where $q$ is an implicit parameter of $\ngam$.
We define $\ngam^{\trans{\Eflag\Bfull}{\Afull}}$ to be
the composite channel
\begin{equation}\label{tildeGammaDef}
    \ngam^{\trans{\Eflag\Bfull} {\Afull}}
    :=
    \E^{\trans{\Eflag\Bfull}{\Bfull}}_{\kappa}\circ \gam^{\trans{\Bfull} {\Afull}}.
\end{equation}
A useful fact regarding compositions is that
\begin{equation}\label{Ic-comp}
  \Ic(\N_1, \rho) \geq \Ic(\N_2\circ \N_1, \rho),
\end{equation}
which is just the quantum data processing inequality for coherent information~\cite{WildeBook}.

We define $\gam^{\trans{\Bfull}{\Afull}}$
by giving its Choi state,
which depends on the parameters $\NN$, $r$, $m$, and $q$.
Defining the composite systems
$\block_{i,j,k} := \A_{i,j,k} \B_{i,j,k}$ and
$\block_k := \set{\block_{i,j,k} : i \in [\NN], j \in [\r]}$,
the Choi state of $\gam^{\trans{\Bfull}{\Afull}}$ is
proportional to
\begin{align}
    \GammaChoi^{\key\shield}&= (|00\>\<00| + |11\>\<11|)^{\key}
    \x \bigotimes_{k=1}^{m} \ob[\Big]{
        \frac{\q}{2}(\omega + \sigma)
    }^{\block_k} \nonumber \\
    &+ (|00\>\<11| + |11\>\<00|)^{\key}
    \x \bigotimes_{k=1}^{m} \ob[\Big]{
        \frac{\q}{2}(\omega - \sigma)
    }^{\block_k} \label{crazystate} \\
    &+ (|01\>\<01| + |10\>\<10|)^{\key}
    \x \bigotimes_{k=1}^{m} \ob[\Big]{
       (\tfrac{1}{2}-\q)\sigma
    }^{\block_k}, \nonumber
\end{align}
where
$\omega^{\block_k}
:= \Ox_{i=1}^{\NN} \Ox_{j=1}^{\r}
\frac{1}{2}(\sap^{\block_{i,j,k}}_{+}+\sap_{-}^{\block_{i,j,k}})$, and
$\sigma^{\block_k} := \Ox_{i=1}^{\NN} \Ox_{j=1}^{\r} \sap_{+}^{\block_{i,j,k}}$
are the Eggeling-Werner data hiding states~\cite{Eggeling_02}.
Here
$\sap_{+}^{\block_{i,j,k}}$
and
$\sap_{-}^{\block_{i,j,k}}$
are
the states proportional to the projectors onto the symmetric
and anti-symmetric subspaces respectively.

In Eq.~(139) of \cite{Horodecki_09b}, a state
$\rho^{rec}_{(p,d,k)}$ is defined. Apart from
$p$, $d$ and $k$, it also implicitly depends
on a parameter $m$, so we will denote it by
$\rho^{rec}_{(p,d,k;m)}$. Our $\HHHOstate^{\key \A \B}$
is precisely $\rho^{rec}_{(q,d,r\NN;m)}$.
From Sections X-A (in particular Lemma 5)
and X-B of \cite{Horodecki_09b} we see that
$\rho^{rec}_{(q,d,r\NN;m)}$ is PPT if
\begin{equation}\label{PPT-params}
0 < q \leq 1/3 \quad\text{and}\quad
\frac{1-q}{q} \geq \left(\frac{d}{d-1}\right)^{r\NN}.
\end{equation}
Since a channel is PPT-binding iff its Choi matrix is PPT,
the same conditions suffice for $\gam$ to be PPT-binding.
This condition is key to our subsequent analysis.

We will now derive from \cite{Horodecki_09b} another important
fact about $\GammaChoi^{\key\shield}$: Defining
\begin{equation}\label{rho1-def}
    \GammaChoi^{\key\A_{1}\B_{1}} :=
\tr_{\A_{2}\B_{2}\cdots\A_{\NN}\B_{\NN}} \HHHOstate^{\key \A \B},
\end{equation}
for an appropriate choices of parameters,
$\rho^{\key\A_{1}\B_{1}}$ can be made arbitrarily close
to a perfect pbit $\gamma^{\key\A_{1}\B_{1}}$
with key $\key$ and shield $\A_{1} \B_{1}$.
In particular, we will use
\begin{lemma}\label{lem:close}
Let $q := 1/3$ and $r := 2m + \ceil{\log_2 m}$.
Then
$\tau := \norm{
            \rho^{\key\A_{1}\B_{1}}
          - \gamma^{\key\A_{1}\B_{1}}
         }_1
         \leq 16 m^{1/2} 2^{-m/4}$
for some perfect pbit $\gamma^{\key\A_{1}\B_{1}}$,
where $\norm{\cdot}_1$ denotes the trace norm.
\end{lemma}

\begin{proof}
    First note that the $\rho^{\key\A_{1}\B_{1}}$
    is simply $\rho^{rec}_{(q,d,r;m)}$.
    Adopting the notation of~\cite{Horodecki_09b},
    let $\norm{A_{0011}}_1$ be the
    norm of the upper right block of the matrix
    $\rho^{rec}_{(q,d,r;m)}$ expanded in the
    computational basis of the key system $\key$.
    In Proposition~4 of~\cite{Horodecki_09b}, it is shown that if
    $1/2 - \norm{A_{0011}}_1 < \epsilon < 1/8$
    then $\tau \leq \delta(\epsilon)$ for some function $\delta(\epsilon)$.
    The function $\delta$ is given in Eq.~(70) of \cite{Horodecki_09b} as
\begin{equation}
  \delta(\epsilon)
  := 2 \of[\big]{8 \sqrt{2 \epsilon} + h(2 \sqrt{2 \epsilon})}^{1/2} + 2 \sqrt{2 \epsilon}
\end{equation}
where $h(x) := -x \log_2 x - (1-x) \log_2 (1-x)$ is the binary entropy function.
Provided $0 \leq x \leq 1/2$,
$h(x)$ is an increasing function of $x$ and
\begin{equation}\label{h-bound}
    h(x) \leq x \log_2 \of*{\frac{1}{x^2}}.
\end{equation}
In particular,
if we assume that
$0 < 2 \sqrt{2 \epsilon} < 1/2$, i.e.,
\begin{equation}
  0 < \epsilon < 1/32,
  \label{eq:Range}
\end{equation}
then
$h(2 \sqrt{2 \epsilon}) \leq \sqrt{8 \epsilon} \log_2 \frac{1}{8 \epsilon}$
and thus
\begin{equation}
  \delta(\epsilon)
  \leq 2 \of*{
    4 \sqrt{8\epsilon}
    + \sqrt{8\epsilon} \log_2 \frac{1}{8 \epsilon}
  }^{1/2} + \sqrt{8\epsilon}.
  \label{eq:Bound1}
\end{equation}
From Eq.~\eqref{eq:Range} we also get
$\log_2 \frac{1}{8 \epsilon} > 1$.
By inserting this extra factor next to $4\sqrt{8\epsilon}$
in Eq.~\eqref{eq:Bound1} we get
\begin{equation}
  \delta(\epsilon)
  \leq 2
  \of[\Big]{ 5 \sqrt{8\epsilon} \log_2 \frac{1}{8\epsilon} }^{1/2} +
  \sqrt{8\epsilon}.
\end{equation}
We can upper bound the last term as $\sqrt{8\epsilon} < (\sqrt{8\epsilon})^{1/2} < (\sqrt{8\epsilon} \log_2 \frac{1}{8\epsilon})^{1/2}$ and the whole expression as
\begin{equation}
  \delta(\epsilon)
  \leq 2^{5/2}
  \of[\Big]{   \sqrt{8\epsilon} \log_2 \frac{1}{8\epsilon} }^{1/2}.
\end{equation}
Thus, we get:
    \begin{equation}\label{tau-bound}
        \tau
        \leq 2^{5/2} \of*{
          \sqrt{8\epsilon}
          \log_2 \frac{1}{8\epsilon}
        }^{1/2}.
    \end{equation}

    Rearranging Eq.~(142) in the proof of Theorem~6
    of~\cite{Horodecki_09b} we find
    $
        1/2 - \norm{A_{0011}}_1 =
        \frac{1}{2} \of[\Big]{
          1 - \frac{(1-2^{-r})^m}
          {1+\of*{\frac{1-2q}{2q}}^m}
        }
    $.
By omitting the factor $1/2$ we get:
\begin{align}
        \frac{1}{2} - \norm{A_{0011}}_1
        &=
        \frac{1}{2} \of[\bigg]{
          1-\frac{(1-2^{-r})^m}
            {1 + \of[\big]{\frac{1}{2q}-1}^m}
        } \\
        &<
          \frac{1 + \of[\big]{\frac{1}{2q}-1}^m - (1-2^{-r})^m}
          {1 + \of[\big]{\frac{1}{2q}-1}^m}.
\end{align}
Setting $q = 1/3$ and using
\begin{equation}
  (1-x)^m \geq 1-mx
\end{equation}
for all $m \in \mathbb{N}$ and $x \in (0,1)$, we have
\begin{align}
    \frac{1}{2} - \norm{A_{0011}}_1
    &< \frac{1 + 2^{-m} - (1-2^{-r})^m}{1 + 2^{-m}} \\
    &< \frac{1 + 2^{-m} - (1-m 2^{-r})}{1 + 2^{-m}} \\
    &= \frac{2^{-m} + m 2^{-r}}{1 + 2^{-m}}
\end{align}
which is a decreasing function of $r$.
Setting $r = 2m + \ceil{\log_2 m}$ we get
\begin{align}
    \frac{1}{2} - \norm{A_{0011}}_1
    &<
    \frac{2^{-m} + m 2^{-(2m + \log_2 m)}}
    {1 + 2^{-m}}\\
    &=
    \frac{2^{-m} + 2^{-2m}}
    {1 + 2^{-m}} = 2^{-m}.
\end{align}

    Therefore,
    for any $m > 5$, and substituting $\epsilon = 2^{-m}$ into \eqref{tau-bound} we obtain
\begin{align}
     \tau
     &\leq 2^{5/2} \of*{
        \sqrt{8 \epsilon}
        \log_2 \frac{1}{8 \epsilon}
      }^{1/2} \\
     &= 2^{5/2} \of[\big]{\sqrt{2^{3-m}}(m-3)}^{1/2}\\
     &\leq 16 \times 2^{-m/4} m^{1/2}
\end{align}
as desired.
\end{proof}

\subsection*{Main result}
The proof of our main result is based on the following two key lemmas. The first proves the coherent information is zero up to $n$ uses of the channel. The second proves that it is non-zero for some larger number of uses, hence the quantum capacity is positive.


\begin{lemma}
\label{lem:converse}
If $\gam$ is PPT-binding,
then for $\kappa\in[0,1]$, $p\in[(1+\kappa^{n})^{-1/n},1]$,
the coherent information of
$n$ uses of the channel $\chan$ is zero:
$
  Q^{(n)}(\chan) = 0.
$
\end{lemma}
\begin{proof}
Using \eqref{switched-coh} and the general fact that
\begin{equation}
    \max_{\rho}
        \Ic(\N\x \mc{M},
        \rho)
        =
    \max_{\rho}
    \Ic(\mc{M}\x \N,
    \rho)
\end{equation}
we have
$
  Q^{(n)}(\chan) = \frac{1}{n}\max_{0 \leq l \leq n} I_l
$, where
\begin{equation}
    I_l := \Ic\of[\big]{
    \ngam^{\x l}
    \x
    \E_p^{\x (n-l)}
    ,
    \rho_{l}^{\Afull^1 \cdots \Afull^n}
    }
    \label{eq:Il}
\end{equation}
and $l$ is the number of switches set to use $\ngam$.
Here $\rho_{l}^{\Afull^1 \cdots \Afull^n}$
is an input state for $n$ uses of the channel that maximises the RHS of \eqref{eq:Il},
where $\Afull^i := \Akey^i \A^i_1 \cdots \A^i_{\NN}$ is the main input system for the $i$-th use of the channel.

From the definition \eqref{tildeGammaDef} and
Eq.~\eqref{flagged-coh},
we see that $I_l$
can be written as a sum of $2^n$ terms,
each corresponding to a possible setting
of the $n$ erasure flags.
We get
\begin{equation}\label{main-ineq-0}
    I_l\leq \kappa^{l} p^{n-l} (-S(\rho_l)) + (1-\kappa^{l}) p^{n-l} \Ic(\gam^{\x l} \x \E_1^{\x n-l}, \rho_l) + (1 - p^{n-l}) S(\rho_l).
\end{equation}
The first term in this bound
is the case where all $n$ channel uses erase
and it follows from \eqref{total-erasure}.
The second term upper bounds the cases where
all of the $\E_p$ uses erase but \emph{not} all
of the $\ngam$ channels do, obtained via \eqref{Ic-comp}.
The final term upper bounds the contribution
from the remaining cases using the trivial bound.

Using \eqref{Ic-comp} and the fact that $\gam$ is PPT-binding,
we obtain
$\Ic(\gam^{\x l} \x \E_1^{\x n - l}, \rho_l)
\leq \Ic(\gam^{\x n}, \rho_l) \leq 0$
and thus we can drop the second term in \eqref{main-ineq-0}:
\begin{equation}
    I_l\leq \of*{-\kappa^l p^{n-l} + 1 - p^{n-l} } S(\rho_l) \leq \of[\big]{1 - (1+\kappa^n) p^n}S(\rho_l),
\end{equation}
where the second inequality follows from $p, \kappa \in [0,1]$.
We find that $I_l \leq 0$ provided that
\begin{equation}\label{p-kappa-n}
  p \geq (1 + \kappa^{n})^{-1/n}.
\end{equation}
On the other hand, $I_l \geq 0$ since we can always choose $\rho_l$ to be pure.
This implies $I_l = 0$ and thus $Q^{(n)}(\chan) = 0$, which completes the proof.
\end{proof}


\begin{lemma}
\label{lem:achiev}
For $p\in (0,1)$, $\kappa \in (0,1/2)$,
we can choose the parameters $q, \NN, r, m, d$
such that the PPT condition \eqref{PPT-params} holds
and $Q^{(\NN+1)}(\chan)>0$.
\end{lemma}
\begin{proof}
Our proof has two parts.
In part (i) we prove a lower bound on $Q^{(\NN+1)}(\chan)$ by analysing
a particular input to the channel. In part (ii) we show that the channel
parameters can be chosen to make this lower bound strictly positive while,
at the same time, satisfying \eqref{PPT-params}.

(i) We number the $\NN+1$ channel uses by $\set{0,1,\dotsc,\NN}$,
and label the systems involved in the $i$-th channel use with
superscript $i$.
The switch systems are set so that the first use of the channel
acts as $\ngam$ on its main input, and the remaining $\NN$ uses
act as $\E_p$.

If $\reg{X}$ and $\reg{Y}$ are two systems of equal dimensions, we use
$\mes^{\reg{XY}} := \proj{\mes}^{\reg{XY}}$
to denote the \emph{maximally entangled state} on $\reg{XY}$ where
$\ket{\mes}^{\reg{XY}} := \sum_{i=0}^{\dim(\reg{X})-1} \ket{i}^{\reg{X}} \ket{i}^{\reg{Y}} / \sqrt{\dim(\reg{X})}$.
Alice prepares maximally entangled states
on subsystems
$\Akey^0 \Akey$ and on
$\A^{0}_i\A^{i}_{1}$
for all $i \in [\NN]$.
The purification of the overall input to the $\NN+1$ uses of the channel $\chan$ is
\begin{equation}
  \ket{\nu} :=
  \ket{0}^{\SWin^0}
  \ket{\mes}^{\Akey \Akey^0}
  \Ox_{i=1}^{\NN}
  \of[\bigg]{
    \ket{1}^{\SWin^i}
    \ket{\alpha}^{\Akey^i}
    \ket{\mes}^{\A^0_i \A^i_1}
    \Ox_{j=2}^{\NN}
    \ket{\beta}^{\A^i_j}
  }
  \label{eq:nu}
\end{equation}
shown in Fig.~\ref{fig:achievability},
where $\ket{\alpha}$ and $\ket{\beta}$ are arbitrary pure states,
$\Akey$ is a reference system, and $\SWin^i$ and $\Afull^i$ are the switch and
main input systems for the $i$-th use of $\chan$, respectively.

The switch settings cause the first use of $\chan$ to act as $\ngam$
on $\Afull^0 = \Akey^0 \A_1^0 \cdots \A_{\NN}^0$
(see \eqref{tildeGammaDef}).
With probability $\kappa$, 
$\ngam$ erases, yielding
$\ketbra{1}{1}^{\Eflag^0} \x \mms^{\Bfull^0}$.
With probability $1-\kappa$, it
sets the erasure flag to $\ketbra{0}{0}^{\Eflag^0}$
and acts as $\gam$ on $\Afull^0$, producing $\Bfull^0$.
At this point the state of
$\Akey\Bfull^0\A_1^1 \cdots \A_1^{\NN}$
is just the Choi state $\GammaChoi^{\Akey\Bkey\A\B}$
defined in \eqref{crazystate} with its
systems relabelled as follows:
$\Bfull \to \Bfull^{0}$ and
$\A_{j} \to \A^{j}_{1}$ for all $j \in [\NN]$.
The switches are set so that the remaining $\NN$ uses of $\chan$
apply $\E_p$ to the each of the systems $\Afull^j$ for each $j \in [N]$.
Bob now applies a simple post-processing operation
to the output systems of all $\NN+1$ channel uses
to obtain a state of a system
$\Bkey \A'_1 \B_1^{} \Gflag \Eflag^{0}$:
He first measures the erasure flags $\Eflag^1 \dotsb \Eflag^{\NN}$.
With probability $1 - p^{\NN}$,
at least one of these flags, say $\Eflag^{j}$,
will be in the state $\proj{0}^{\Eflag^{j}}$,
and the state of $\A_1^{j}$
has been perfectly transferred to his system $\B_1^{j}$.
Otherwise, with probability $p^{\NN}$, Bob
picks an arbitrary $j \in [\NN]$. In this case the
state of $\Eflag^{j}$ is $\proj{1}^{\Eflag^{j}}$
and the state of $\B_1^{j}$ is maximally mixed
and uncorrelated with any other system.
Now, as depicted in Fig.~\ref{fig:achievability},
Bob transfers the state of $\Eflag^{j}$
to a system $\Gflag$, the state of $\B_1^{j}$
to $\A'_1$, the state of $\B^{0}_j$ to $\B_1$
and $\Bkey^{0}$ to $\Bkey$.
Bob then discards all of his systems except for
$\Bkey \A'_1 \B_1^{} \Gflag \Eflag^{0}$,
which are now in the state
\begin{align}
    \eta^{ \key \A'_1 \B_1^{} \Gflag \Eflag^{0} } :&=
    \kappa \mms^{\Akey}
            \x \sigma^{\Bkey \A'_1 \B_1^{} \Gflag}
           \x \proj{1}^{\Eflag^{0}} \notag\\
    &{}+(1-\kappa) p^{\NN}
            \HHHOstate^{\key\B_1^{}}
            \x \mms^{\A_1'}
            \x \proj{1}^{\Gflag}
           \x \proj{0}^{\Eflag^{0}} \label{eta-def}\\
    &{}+(1-\kappa) (1-p^{\NN})
            \HHHOstate^{\key \A'_1 \B_1^{}}
            \x \proj{0}^{\Gflag}
           \x \proj{0}^{\Eflag^0}.\notag
\end{align}
Here
$\HHHOstate^{\key \A'_1 \B_1^{}}
:= \id^{\trans{\A'_1}{\A_1^{}}}
(\HHHOstate^{\key \A_1 \B_1})$,
where $\HHHOstate^{\key \A_1 \B_1}$
is the state \eqref{rho1-def} from Lemma \ref{lem:close}.
The details of $\sigma^{\Bkey \A'_1 \B_1^{} \Gflag}$ are unimportant.
The first term in \eqref{eta-def} corresponds to the case where
the first channel use
erases. When the first use doesn't erase, the case where all other uses
erase yields the second term, and the case where at least one of the other
uses does not erase gives the third term.

Let us call Bob's post-processing operation $\mc{P}$.
Using the state $\nu := \proj{\nu}$ from \eqref{eq:nu}, we can write
\begin{equation}
    (\NN+1)Q^{(\NN+1)}(\chan)\geq \Ic(\chan^{\x \NN+1}, \nu)\geq \Ic(\mc{P} \circ \chan^{\x \NN+1}, \nu)= I(\Akey \rangle \Bkey \A'_1 \B_1^{} \Gflag \Eflag^0)_{\eta^{\key \A'_1 \B_1^{} \Gflag \Eflag^0}},
\end{equation}
where the composition property \eqref{Ic-comp} was used.
Given the ``flagged'' structure of \eqref{eta-def},
we can use \eqref{flagged-state-coh}:
\begin{equation}
    (\NN+1)Q^{(\NN+1)}(\chan)
    \geq{}
    \kappa I(\Akey \rangle \Bkey \A'_1 \B_1^{} \Gflag)_
    {\mms^{\Akey} \x \sigma^{\Bkey \A'_1 \B_1^{} \Gflag}}{}+ (1-\kappa) p^{\NN} I(\Akey \rangle \Bkey \A'_1 \B_1)_
    {\HHHOstate^{\key\B_1^{}}
            \x \mms^{\A_1'}}{}+ (1-\kappa) (1-p^{\NN})I(\Akey \rangle \Bkey \A'_1 \B_1^{})_{\HHHOstate^{\key \A'_1 \B_1^{}}}.
\end{equation}
The first term is $-\kappa S(\mms^\Akey) = -\kappa$.
If
$
  \tau = \norm{
    \HHHOstate^{\key\A'_1\B_1}
  - \gamma^{\key\A'_1\B_1}
  }_1
$
for some perfect pbit $\gamma^{\key\A'_1\B_1}$,
then by the monotonicity of the trace distance under CPTP maps
\begin{equation}\label{monotonicity}
  \tau \geq \norm{
    \HHHOstate^{\key\B_1} \x \mms^{\A'_1}
  - \gamma^{\key\B_1} \x \mms^{\A'_1}
  }_1.
\end{equation}
In what follows, we will use the Alicki-Fannes inequality~\cite{Alicki_04}.
This states that for $\rho^{\reg{RB}}$ and $\sigma^{\reg{RB}}$ such that
$\tau := \norm{\rho^{\reg{RB}} - \sigma^{\reg{RB}}}_1 < 1$ we get
\begin{equation}
  \abs*{
    I(\reg{R} \rangle \reg{B})_{\rho^{\reg{RB}}}
  - I(\reg{R} \rangle \reg{B})_{\sigma^{\reg{RB}}}
  } \leq 4 \tau \log_2 \dim(\reg{R}) + 2 h(\tau).
\end{equation}
Using~\eqref{monotonicity} and properties
\eqref{pbit-prop-1}, \eqref{pbit-prop-2}, $\dim(\Akey) = 2$
together with the Alicki-Fannes inequality we have
\begin{align}
 &I(\Akey \rangle \Bkey \A'_1 \B_1)_{\HHHOstate^{\key \A'_1 \B_1}}
  \geq 1 - \Delta, \\
 &I(\Akey \rangle \Bkey \A'_1 \B_1)_{\HHHOstate^{\key\B_1^{}}
            \x \mms^{\A_1'}}
  \geq - \Delta,
\end{align}
where
\begin{equation}
    \Delta := 4 \tau + 2 h(\tau).
    \label{DeltaEqn}
\end{equation}
Therefore,
$
    (\NN+1)Q^{(\NN+1)}(\chan)
    \geq (1-\kappa)(1-p^{\NN} - \Delta) - \kappa
$
which is strictly positive if
\begin{equation}\label{DeltaRequirement}
    \Delta < 1-p^{\NN} - \frac{\kappa}{1-\kappa}.
\end{equation}

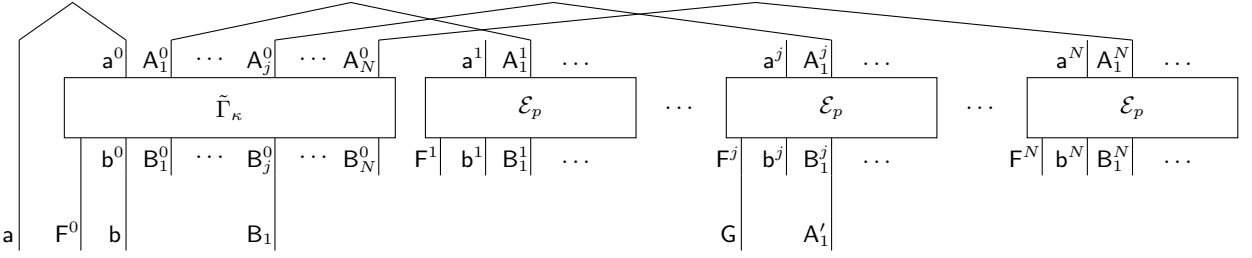
\begin{figure*}[t]
\centering

\begin{tikzpicture}[lb/.style = {font = \small, anchor = east, text height = 1.5ex, text depth = 0.25ex}]
  \def\W{4}       
  \def\h{0.9}     
  \def\w{0.15*\W} 
  \def\H{1}       

  \def\dts{$\dots$} 

  \newcommand{\Ggate}[3]{
    \path (#1*\W,0) coordinate (G#2);           
    \foreach \i in {-0.4, 1.9} {
      \path (G#2)+(\i*\w,-\h+0.25) node {\dts}; 
      \path (G#2)+(\i*\w, \h-0.25) node {\dts}; 
    }
    \foreach \i/\n in {-2.3/0, -1.3/1, 1.0/j, 3.3/N} {
      \draw (G#2)+(\i*\w, \h) coordinate (A#2-\n) --
                 +(\i*\w,-\h) coordinate (B#2-\n);
    }
    \path (A#2-0)+(0.10,-0.28) node[lb] {$\Akey^{#2}$};
    \path (A#2-1)+(0.10,-0.28) node[lb] {$\Ashield^{#2}_1$};
    \path (A#2-j)+(0.10,-0.28) node[lb] {$\Ashield^{#2}_j$};
    \path (A#2-N)+(0.10,-0.28) node[lb] {$\Ashield^{#2}_N$};
    \path (B#2-0)+(0.10, 0.20) node[lb] {$\Bkey^{#2}$};
    \path (B#2-1)+(0.10, 0.20) node[lb] {$\Bshield^{#2}_1$};
    \path (B#2-j)+(0.10, 0.20) node[lb] {$\Bshield^{#2}_j$};
    \path (B#2-N)+(0.10, 0.20) node[lb] {$\Bshield^{#2}_N$};
    \draw (G#2)++(-3.3*\w,0) -- +(0,-\h-\H) coordinate (B#2-E);
    \path (B#2-E)+(0.10, 0.20) node[lb] {$\Eflag^{#2}$};
    \node[draw, fill = white, minimum width = \W*1.1cm, minimum height = 0.8cm] at (G#2) {$#3$};
  }

  \newcommand{\Egate}[3]{
    \path (#1*\W,0) coordinate (G#2);              
    \path (G#2)+(1.6*\w,-\h+0.25) node[lb] {\dts}; 
    \path (G#2)+(1.6*\w, \h-0.25) node[lb] {\dts}; 
    \foreach \i/\n in {-1/0, 0/1} {
      \draw (G#2)+(\i*\w, \h) coordinate (A#2-\n) --
                 +(\i*\w,-\h) coordinate (B#2-\n);
    }
    \path (A#2-0)+(0.10,-0.28) node[lb] {$\Akey^{#2}$};
    \path (B#2-0)+(0.10, 0.20) node[lb] {$\Bkey^{#2}$};
    \path (A#2-1)+(0.10,-0.28) node[lb] {$\Ashield^{#2}_1$};
    \path (B#2-1)+(0.10, 0.20) node[lb] {$\Bshield^{#2}_1$};
    \draw (G#2)+(-2*\w,0) -- +(-2*\w,-\h) coordinate (B#2-E);
    \path (B#2-E)+(0.10, 0.20) node[lb] {$\Eflag^{#2}$};
    \node[draw, fill = white, minimum width = \W*0.7cm, minimum height = 0.8cm] at (G#2) {$#3$};
  }

  \Ggate{0.0}{0}{\tilde{\Gamma}_\kappa}
  \Egate{1.0}{1}{\E_p}                     \path (1.5*\W,0) node {\dts};
  \Egate{2.0}{j}{\E_p}                     \path (2.5*\W,0) node {\dts};
  \Egate{3.0}{N}{\E_p}

  \path (-0.7*\W, \h) coordinate (AR);
  \path (-0.7*\W,-\h) coordinate (BR);

  \def\r{0.5} 
  \path ($(A0-0)!0.5!(AR)$)+(0,\r) coordinate (S0);
  \path ($(A0-1)!0.5!(A1-1)$)+(0,\r) coordinate (S1);
  \path ($(A0-j)!0.5!(Aj-1)$)+(0,\r) coordinate (Sj);
  \path ($(A0-N)!0.5!(AN-1)$)+(0,\r) coordinate (SN);
  \draw (A0-0) -- (S0) -- (AR);
  \draw (A0-1) -- (S1) -- (A1-1);
  \draw (A0-j) -- (Sj) -- (Aj-1);
  \draw (A0-N) -- (SN) -- (AN-1);

  \draw (AR) -- (BR)
               -- ++(0,-\H) +(0.05, 0.20) node[lb] {$\Akey$};
  \draw (B0-0) -- ++(0,-\H) +(0.05, 0.20) node[lb] {$\Bkey$};
  \draw (B0-j) -- ++(0,-\H) +(0.10, 0.20) node[lb] {$\Bshield_1$};
  \draw (Bj-E) -- ++(0,-\H) +(0.05, 0.20) node[lb] {$\Gflag$};
  \draw (Bj-1) -- ++(0,-\H) +(0.10, 0.20) node[lb] {$\Ashield'_1$};
\end{tikzpicture}
\caption{Input that achieves positive coherent information for $\NN+1$ uses of $\chan$.
The switch systems are not shown for clarity. They have been given (pure)
inputs such that the first channel acts as $\ngam$
and the rest act as $\E_p$ on their main inputs.}
\label{fig:achievability}
\end{figure*}

(ii) We will now show how the parameters must be chosen.
First, to ensure that \eqref{PPT-params} is satisfied,
we specify that $d := 2 \NN r$ and $q := 1/3$.
Now, if $\kappa \in (0,1/2)$
then $\kappa/(1-\kappa) \in (0,1)$,
so for any $p \in (0,1)$
we can always choose $\NN$ large enough to make the RHS
of \eqref{DeltaRequirement} positive.
Fixing this value of $\NN$,
we then must choose $m$ and $r$
to make $\Delta$ small enough to
satisfy \eqref{DeltaRequirement}.
Lemma \ref{lem:close} tells us that with $q = 1/3$
and $r = 2 m + \log_2 m$, we have
 $   \tau \leq 16 m^{1/2} 2^{-m/4}$.
   Recall that
    $\Delta = 4\tau + 2 h(\tau)$,
    into which we are substituting
    $
        \tau \leq 16 m^{1/2} 2^{-m/4}
    $.
    Provided $0 \leq x \leq 1/2$,
    $h(x)$ is an increasing function of $x$,
    and $h(x) \leq 2 x \log_2 \frac{1}{x}$,
    so $h(\tau) \leq h(16 m^{1/2} 2^{-m/4}) \leq
    4 m^{3/2} 2^{-m/4}$ (provided
    $16 m^{1/2} 2^{-m/4} \leq 1/2$).
    We get
    \begin{align}\label{eq:Delta}
        \Delta
       \leq 64 m^{1/2} 2^{-m/4}
        + 8 m^{3/2} 2^{-m/4}
       \leq 72 \times 2^{-m/4} m^{3/2}.
    \end{align}
One can choose $m$ to make this as small as required.
\end{proof}

We can now prove our main result:
\begin{theorem}
Let $\chan$ be the channel defined in Eq.~\eqref{eq:ourchannel_again}.
For any positive integer $n$, if $\kappa \in (0,1/2)$ and $p \in [(1+\kappa^{n})^{-1/n},1]$ then we can choose $\q, d, r,\NN,\m$ such that:
\begin{enumerate}
  \item $Q^{(n)}(\chan) = 0$ and
  \item $Q^{(\NN+1)}(\chan) > 0$, and therefore $Q(\chan) > 0$.
\end{enumerate}
\end{theorem}


\begin{proof}
  In Lemma~\ref{lem:converse} we show that if $\gam$ is PPT-binding and $\kappa$, $p$ satisfy $\kappa \in [0,1]$ and $p \in [(1+\kappa^{n})^{-1/n},1]$, then the first statement holds. In Lemma~\ref{lem:achiev} we show that for any $\kappa \in (0,1/2)$ and $p \in (0,1)$, we can choose the parameters $\q, d, r,\NN,\m$ so that the second statement holds \emph{and} $\gam$ is PPT-binding. Therefore, for $(\kappa, p)$ in the intersection of the two regions, the channel $\chan$ satisfies both statements.
\end{proof}

To be concrete, we can choose $\kappa = 1/4$ (so that $\kappa/(1-\kappa) = 1/3$) and choose $p = (1+\kappa^n)^{-1/n}$. We can then choose $\NN$ so that $1 - p^{\NN} \geq 2/3$: we require $(1+\kappa^n)^{-\NN/n} < 1/3$. Taking logs of both sides, rearranging and using $x/\ln(2) \geq \log_2 (1+x)$, we have
$\NN > (\log_2 3)(\ln 2)n 4^n$, so let us take $\NN = 2 n 4^n$. We must now choose $m$ large enough ($m \geq 68$) that $\Delta < 1/3$ in \eqref{eq:Delta}.

\end{widetext}

\end{document}